\documentclass[conference,twoside]{ieeeconf}
\usepackage{cite}
\usepackage{amsmath,amssymb,amsfonts}
\usepackage{algorithmic}
\usepackage{graphicx}
\usepackage{algorithm,algorithmic}
\usepackage{hyperref}
\hypersetup{hidelinks}

\usepackage[capitalise]{cleveref}
\crefname{assumption}{assumption}{assumptions}
\crefname{lemma}{lemma}{lemmas}

\DeclareMathOperator{\co}{conv}
\DeclareMathOperator{\interior}{int}
\DeclareMathOperator{\kernel}{Ker}

\newcommand{\D}{\mathrm{d}}
\DeclareMathOperator{\image}{Im}
\DeclareMathOperator{\proj}{Pr}
\newcommand\xqed[1]{%
  \leavevmode\unskip\penalty9999 \hbox{}\nobreak\hfill
  \quad\hbox{#1}}
\newcommand\demo{\xqed{$\triangle$}}

\newtheorem{remark}{Remark}
\newtheorem{assumption}{Assumption}
\newtheorem{proposition}{Proposition}
\newtheorem{lemma}{Lemma}
\newtheorem{theorem}{Theorem}
\newtheorem{corollary}{Corollary}

\usepackage{mathtools}
\let\labelindent\relax
\usepackage{enumitem}

\usepackage{textcomp}
\def\BibTeX{{\rm B\kern-.05em{\sc i\kern-.025em b}\kern-.08em
    T\kern-.1667em\lower.7ex\hbox{E}\kern-.125emX}}
    
\begin{document}
\title{Multiple Model Reference Adaptive Control with Blending for Non-Square Multivariable Systems}
\author{Alex Lovi, Bar\i{\c{s}} Fidan, and Christopher Nielsen
\thanks{This work is supported by the Natural Sciences and Engineering Research Council of Canada (NSERC). This work has been submitted to the IEEE for possible publication. Copyright may be transferred without notice, after which this version may no longer be accessible.}
\thanks{Department of Electrical and Computer Engineering, University of Waterloo, Waterloo ON, N2L 3G1 Canada (e-mail: alovi@uwaterloo.ca).}
\thanks{Department of Electrical and Computer Engineering, University of Waterloo, Waterloo ON, N2L 3G1 Canada (e-mail: cnielsen@uwaterloo.ca).}
\thanks{Department of Mechanical and Mechatronics Engineering, University of Waterloo, Waterloo ON, N2L 3G1 Canada (e-mail: fidan@uwaterloo.ca).}
}

\maketitle

\begin{abstract}
    In this paper we develop a multiple model reference adaptive controller (MMRAC) with blending. The systems under consideration are non-square, i.e., the number of inputs is not equal to the number of states; multi-input, linear, time-invariant with uncertain parameters that lie inside of a known, compact, and convex set. Moreover, the full state of the plant is available for feedback. A multiple model online identification scheme for the plant's state and input matrices is developed that guarantees the estimated parameters converge to the underlying plant model under the assumption of persistence of excitation. Using an exact matching condition, the parameter estimates are used in a control law such that the plant's states asymptotically track the reference signal generated by a state-space model reference. The control architecture is proven to provide boundedness of all closed-loop signals and to asymptotically drive the state tracking error to zero. Numerical simulations illustrate the stability and efficacy of the proposed MMRAC scheme.
\end{abstract}


\section{Introduction}



The use of multiple models to describe the dynamics of uncertain systems has been studied~\cite{narendra_new_2012, morse_simple_2011, kuipers_multiple_2010, hespanha_overcoming_2003}, and shown to improve the transient-time performance~\cite{PrasadG.Maruthi2020}, the steady-state tracking performance~\cite{hespanha_overcoming_2003,VuLinh2011SCoU}, and robustness~\cite{hespanha_switching_2002} when compared to a single model approach. Multiple model control techniques can be broadly divided into \emph{switching} control, which allows for the selection of a best model in different dynamic situations~\cite{hespanha_multiple_2001, hespanha_switching_2002}, and \emph{blending} control, where the information from different models is combined to get a single description of the system~\cite{zengin_blending_2021, buyukkabasakal_mixing_2017, kuipers_multiple_2010,mancilla-aguilar_algorithm_2015}. In this paper we opt for a blending technique since it allows for better closed-loop performance when compared to any single model technique~\cite{EvensenGeir2009DATE}, and avoids possible undesirable transient-time behavior that switching control may exhibit~\cite{DehghaniArvin2007UacA,BaldiSimone2010Muas}.




Multiple model approaches with blending have various promising applications in identification and control of time-varying (TV) systems, including adaptive identification of MIMO, linear, periodically TV systems (with known period)~\cite{narendra_adaptive_2019} and improvement of the transients and adaptation speed in adaptive control of uncertain, TV, MIMO systems~\cite{zengin_blending_2021}. The stability and robustness properties were studied in~\cite{kuipers_multiple_2010}. Mixing adaptive techniques have also been used to achieve faster tracking for a class of nonlinear discrete-time systems~\cite{ZhangYan-Qi2023Ammc}. For the case of experimental results, the use of mixing control has allowed to decrease overshoots, settling time, steady state error~\cite{PandeyVinayKumar2017Cdfa,DuttaLakshmi2021Ampc}, and designing fault tolerant controllers~\cite{buyukkabasakal_mixing_2017}. Other applications include multiple model estimation of power systems models~\cite{moffat_multiple_2021}, development of an adaptive controller for partially-observed Boolean dynamical systems~\cite{imani_multiple_2017}, and distributed state estimation using a network of local sensors~\cite{wang_fully_2016}.

Model reference control has been rigorously and methodically studied for many years and proposes promising applications with detailed design procedures.
When we consider systems with uncertainty, model reference adaptive controllers (MRAC) are a versatile technique that achieves robust and satisfactorily closed-loop performances~\cite{ioannou_adaptive_2006,tao2003adaptive}. The combination of MRAC with multiple model approaches to obtain a continuous input signal calculated using all the identification errors from all the models is studied in~\cite{narendra_new_2012,zhuo_han_new_2012}. Similar work is present in~\cite{ahmadian_new_2015}, where adaptive identification models are considered instead of fixed models. In~\cite{zengin_blending_2021}, a similar identification scheme is integrated with linear quadratic optimal controllers to design a multiple MRAC (MMRAC) scheme for tracking reference signals generated by a linear reference model. The asymptotic tracking of the MMRAC scheme for linear, time-invariant (LTI), MIMO systems is proven and simulations are presented to evaluate and validate the performance.

In this paper we consider non-square, multi-input, LTI systems with polytopic parameter uncertainty, whose unknown plant models are in the interior of the convex hull of a finite number of fixed models. Based on full state measurements, we develop a parameter identification scheme that estimates a weight vector which determines the convex combination that yields the state-space representation of the uncertain system.

This paper extends our previous work~\cite{lovi2022} in several significant ways. The systems under consideration do not have to be square, i.e., the number of inputs can be less or equal to the number of states. Moreover, unlike in~\cite{lovi2022}, in this article the number of fixed models that define the polytopic uncertainty of the plant can be arbitrary. This relaxation allows us to develop a procedure to obtain a set of fixed corner systems to use in the identification process. Furthermore, we provide sufficient conditions under which the parameter estimates asymptotically converge to the plant's uncertain parameters.
The simulations presented here point to faster convergence of the parameter estimates to their true values, making the state error of the proposed MMRAC scheme also converge faster compared to a single model MRAC approach.

The problem formulation is presented in \Cref{sec:Problem_Formulation}, together with the assumptions. In~\Cref{sec:cornermodel}, we present a selection process for the corner systems. In \Cref{sec:sysidentification}, we present the identification scheme with its stability analysis. We use these results in \Cref{sec:syscontrol} to develop the MMRAC scheme. A set of simulations are presented comparing the MMRAC scheme to the single model case for uncertain systems tracking references in \Cref{sec:Simulation} using MATLAB and Simulink based simulations. We finalize with the conclusions in \Cref{sec:Conclusions}.

\subsubsection*{Notation} If $ S \subset \mathbb{R}^{n} $, then $ \co \left( S \right)$ denotes the closed convex hull of $ S $, the interior of $S$ is written $ \interior (S) $. Let $ \| \cdot \| $ denote the 2-norm of both vectors and matrices. The kernel of a matrix $A$ is denoted by $ \kernel A $. Given two matrices $ Q, R \in \mathbb{R}^{n \times n} $, we write $ Q \succeq R $ if $ Q - R $ is positive semi-definite. For two signals $ f_1 $ and $ f_2 $ we write $ f_1 \equiv f_2 $, if there exist $ \alpha$, $\beta > 0 $ such that for all $ t \geq 0 $, $ {\| f_1 \left( t \right) - f_2 \left( t \right) \| \leq \beta e^{-\alpha t}} $. Further, if $e_f=f_1-f_2$ satisfies $\dot{e}_f(t)=-\alpha e_f(t)$ then we write $ f_1 \stackrel{\alpha}{=} f_2 $. If $ x \in \mathbb{R}^N $, then $ x_i $ denotes the $ i $-th component, and $ {\bar{x} = \left( x_1, \cdots, x_{N-1}  \right) \in \mathbb{R}^{N-1}} $. 
A signal $\Phi : [0,\infty) \rightarrow \mathbb{R}^{k}$ is persistently exciting (PE) if there exist $\alpha_{\Phi 1},\alpha_{\Phi 2},T_\Phi>0$ such that for every $t\geq 0$
\begin{equation}
    \alpha_{\Phi 1} \mathbb{I} \preceq \int_{t}^{t+T_\Phi} \Phi(\tau) \Phi^\top(\tau) \D \tau \preceq \alpha_{\Phi 2} \mathbb{I}.
    \label{eq:PEinequality}
\end{equation}
\section{Problem Formulation}
\label{sec:Problem_Formulation}

Consider the MIMO, LTI system
\begin{equation}
        \dot{x}_p \left( t \right) = A_p x_p \left( t \right) + B_p u \left( t \right),\: x_p \left( 0 \right) = x_{p0},\: t \geq 0,
	\label{pr:plant}
\end{equation}
where $ A_p \in \mathbb{R}^{n \times n} $ and $ B_p \in \mathbb{R}^{n \times m} $ are unknown constant matrices, $ B_p $ is full column rank, $ x_p \left( t \right) \in \mathbb{R}^{n} $, $ u \left( t \right) \in \mathbb{R}^{m} $ are the state of the system and the control input, respectively. We assume that $ x_p(t)$ is available for feedback and that the number states $n$ and inputs $m$ are known.
\begin{remark}
    The results of this article can be applied to systems with structured uncertainties of the form $ {\left( A_p \left( \eta \right), B_p \left( \eta \right) \right)} $, where $ {A_p \left( \eta \right) \in \mathbb{R}^{n \times n}} $ and $ {B_p \left( \eta \right) \in \mathbb{R}^{n \times m}} $ are unknown constant matrices dependent on the uncertain parameter vector $ {\eta \in S \subset \mathbb{R}^{q}} $, belonging to a compact set. \demo
    \label{remarkParameterization}
\end{remark}

This article aims to design a controller such that the state $ x_p $ of the plant \eqref{pr:plant} asymptotically tracks the signal $ x_r $ generated by the reference model
\begin{equation}
	\dot{x}_r \left( t \right) = A_r x_r \left( t \right) + B_r r \left( t \right),\: x_r \left( 0 \right) = x_{r0},\: t \geq 0,
	\label{pr:model_reference}
\end{equation}
where $A_r \in \mathbb{R}^{n \times n} $ and $ B_r \in \mathbb{R}^{n \times m} $ are known constant matrices, $A_r$ is Hurwitz, and $ r : \left[ 0,\infty \right) \rightarrow \mathbb{R}^m $ is a known, bounded, piecewise continuous reference signal. It is assumed that the plant \eqref{pr:plant} and the reference model \eqref{pr:model_reference} satisfy the exact matching conditions~\cite{tao2003adaptive}, as stated in the following assumption.
\begin{assumption}
    There exist matrices $ K^{*} \in \mathbb{R}^{m \times n} $ and $ {L^{*} \in \mathbb{R}^{m \times m}} $ such that
	\begin{align}
	    A_p + B_p K^{*} &= A_r, \label{pr:realKstar} \\
	    B_p L^{*} &= B_r. \label{pr:realLstar}
	\end{align} \demo
	\label{pr:assumption_geometric}
\end{assumption}
\Cref{pr:assumption_geometric} is a necessary and sufficient condition to guarantee the existence of a solution to the tracking problem when the plant's model~\eqref{pr:plant} is known~\cite{tao2003adaptive}. In addition,~\eqref{pr:realLstar} is equivalent to $\image B_r \subseteq \image B_p$, which means that, without loss of generality, we can assume that $\image B_p = \image B_r$.


Our approach to solving the aforementioned control design problem involves online parameter identification of the system matrices $\begin{bmatrix}
    A_p & B_p
\end{bmatrix} \eqqcolon \Theta_p$ by defining a compact, convex, uncertainty polytope such that for every extreme point of the polytope, also referred to as corner, there is a fixed model
\begin{align*}
    \dot{x}_{i} (t) &= A_i x_{i} (t) + B_i u(t),
\end{align*}
with system matrices $\begin{bmatrix}
    A_i & B_i
\end{bmatrix} \eqqcolon \Theta_i$, $i \in \left\{1,\cdots,N\right\}$. We define the set
\begin{equation}
    \mathcal{S} \coloneqq \left\lbrace
    \begin{bmatrix}
        A_1 & B_1
    \end{bmatrix}, 
    \cdots, 
    \begin{bmatrix}
        A_N & B_N
    \end{bmatrix} 
    \right\rbrace
    \label{pr:setS}
\end{equation}
consisting of every corner model.
\begin{assumption}
	There exist a finite set $\mathcal{S}$ of known system matrices $ \begin{bmatrix} A_i & B_i \end{bmatrix} \in \mathbb{R}^{n\times (n+m)} $, such that
    \begin{enumerate}
        \item[i)] $ {\begin{bmatrix} A_p & B_p \end{bmatrix} \in \interior \left( \co \left( \mathcal{S} \right) \right)} $. 
        \item[ii)] Every convex combination of the $B_i$ is full column rank.
    \end{enumerate}
    \demo
	\label{pr:assumption_convex_hull}
\end{assumption}


The first item of \Cref{pr:assumption_convex_hull} can be achieved from a system identification process which we describe in~\Cref{subs:obstainingS}. We require~\Cref{pr:assumption_convex_hull}~(ii) to ensure that we do not have redundant inputs, or we do not drop rank of the number of inputs, losing control authority of the system. Given a set $\mathcal{S}$, there exist numerical methods to verify whether every convex combinations of the $B_i$'s is full rank~\cite{KolodziejczakBarbara1999Ccom}. 




Every point of the polytope $\co \left( \mathcal{S} \right) $ can be expressed as a convex combination of the corner models $ \begin{bmatrix} A_i & B_i \end{bmatrix} $, which implies that the following set is non-empty:

\footnotesize
    \begin{equation}
        \mathcal{W} \coloneqq \left\lbrace
            w \in \left[ 0,1 \right]^N:
            \begin{bmatrix}
                A_p & B_p
            \end{bmatrix}
            = \sum_{i = 1}^{N} w_i
            \begin{bmatrix}
                A_i & B_i
            \end{bmatrix}
            , \sum_{i=1}^{N} w_i = 1
        \right\rbrace.
        \label{pr:convexity}
    \end{equation}
\normalsize
Consequently, the problem of identifying the unknown system matrix $ \begin{bmatrix} A_p & B_p \end{bmatrix} $ is equivalent to the problem of identifying a vector $w \in \mathcal{W}$. A preliminary study of the problem with $ N = n + 1 $, and $ m = n $ was studied in \cite{lovi2022}. In this paper unlike in~\cite{lovi2022}, we study the problem for any number of inputs $m>0$, and an arbitrary number $N$ of corner models.
\section{Corner Model Selection}
\label{sec:cornermodel}


Consider a set $\mathcal{S}$ that satisfies~\Cref{pr:assumption_convex_hull} and has $N$ elements. In this section we develop a constructive procedure that, starting with $\mathcal{S}$, produces a new set
\begin{align}
\label{eq:Sprime}
    \mathcal{S}^\prime = \left\lbrace
    \begin{bmatrix}
        A_1^\prime & B_1^\prime
    \end{bmatrix}, 
    \cdots, 
    \begin{bmatrix}
        A^\prime_{N^\prime} & B^\prime_{N^\prime}
    \end{bmatrix} 
    \right\rbrace,
\end{align}
that also satisfies~\Cref{pr:assumption_convex_hull}, and such that for every $i \in \left\{1,\cdots,N\right\}$ there exist matrices $ K_i \in \mathbb{R}^{m \times n} $ and $ {L_i \in \mathbb{R}^{m \times m}} $ such that,
\begin{align}
    A_i^\prime + B_i^\prime K_i &= A_r, \label{pr:fixedKi} \\
    B_i^\prime L_i &= B_r. \label{pr:fixedLi}
\end{align}



\subsection{Satisfying the Matching Conditions for the Corner Models}


In the next proposition we show how to use the information of~\Cref{pr:assumption_geometric} and a set $\mathcal{S}$ that satisfies~\Cref{pr:assumption_convex_hull} to obtain a new set $\mathcal{S}^\prime$ such that for every element of $\mathcal{S}^\prime$ there exist matrices $ K_i \in \mathbb{R}^{m \times n} $ and $ {L_i \in \mathbb{R}^{m \times m}} $ that satisfy \eqref{pr:fixedKi} and \eqref{pr:fixedLi}.




\begin{proposition}
    Suppose that the plant~\eqref{pr:plant} and reference model~\eqref{pr:model_reference} are such that~\Cref{pr:assumption_geometric} is satisfied. If there exists a set $\mathcal{S}$ that satisfies~\Cref{pr:assumption_convex_hull}, then there exists a set $\mathcal{S}^\prime$ that also satisfies~\Cref{pr:assumption_convex_hull}, and furthermore
    \begin{enumerate}
        \item $ \co (\mathcal{S}^\prime) \subseteq \co (\mathcal{S})$.
        \item For each $i \in \left\{1,\cdots,N\right\}$, there exist matrices $ K_i \in \mathbb{R}^{m \times n} $ and $ {L_i \in \mathbb{R}^{m \times m}} $ such that \eqref{pr:fixedKi} and \eqref{pr:fixedLi} are satisfied.
    \end{enumerate}
    \label{propCornerModel}
\end{proposition}

\begin{proof}
    Let $(A_p, B_p)$ and $(A_r, B_r)$ be the matrix pairs that represent a plant, and a reference model, respectively. Assume that they satisfy~\Cref{pr:assumption_geometric}. Moreover, assume that there exists a set $\mathcal{S}$ with $N$ elements that satisfies~\Cref{pr:assumption_convex_hull}. The first, and trivial case, is if for every $i=1,\cdots,N$, there exist matrices $ K_i \in \mathbb{R}^{m \times n} $ and $ {L_i \in \mathbb{R}^{m \times m}} $ such that every corner model 
    $[A^\prime_{i}~~B^\prime_{i}] := [A_{i}~~B_{i}]$ satisfies \eqref{pr:fixedKi} and \eqref{pr:fixedLi}. In that case, we can define $\mathcal{S}^\prime = \mathcal{S}$, and we have satisfied~\Cref{pr:assumption_geometric,pr:assumption_convex_hull}.


    For the case when there exists some $j\in\{1,\cdots,N\}$ such that
    $[A^\prime_{j}~~B^\prime_{j}] := [A_{j}~~B_{j}] \in \mathcal{S}$ does not satisfy \eqref{pr:fixedKi} or \eqref{pr:fixedLi}, consider the closed convex hull of $\mathcal{S}$, i.e., $\co \left( \mathcal{S} \right).$
    The set $\co \left( \mathcal{S} \right)$ is a compact, convex polytope, and it is non-empty since $\begin{bmatrix}
        A_p & B_p
    \end{bmatrix}$ is an element of $\co \left( \mathcal{S} \right)$ by~\Cref{pr:assumption_convex_hull} (i). Next, take the set of matrices that satisfy~\Cref{pr:realKstar,pr:realLstar}: 
    \begin{align*}
        \mathcal{T} \coloneqq&
        \big\lbrace
            \begin{bmatrix}
                A_x & B_x
            \end{bmatrix}
            \in \mathbb{R}^{n\times(n+m)}: \\
            &A_x + B_x K = A_r, B_x L = B_r
            , L \in \mathbb{R}^{m\times m}, K \in \mathbb{R}^{m\times n} \big\rbrace. 
    \end{align*}
    The set $\mathcal{T}$ is a hyperplane, since it can be written as $n(n+m)$ linear equations (see Section 2.2.1~of~\cite{alma991001125329705160}), therefore it is a convex polytope. From~\Cref{pr:assumption_geometric}, we see that if we replace $K$ and $L$ by $K^*$ and $L^*$, respectively, the element $\begin{bmatrix}
        A_p & B_p
    \end{bmatrix}$ belongs to $\mathcal{T}$, which means $\mathcal{T}$ is non-empty. The intersection set
    \begin{equation*}
        \mathcal{P} \coloneqq \co \left( \mathcal{S} \right) \cap \mathcal{T}.
    \end{equation*}
    is also non-empty, since $\begin{bmatrix}
        A_p & B_p
    \end{bmatrix}$ is an element of both $\co \left( \mathcal{S} \right)$ and $\mathcal{T}$. Moreover, the intersection of a compact, convex polytope and a convex polytope is another compact, convex polytope (see Section 2.3.1~of~\cite{alma991001125329705160}), which in turn implies that we can obtain $N^\prime$ corner models $\begin{bmatrix}
        A_i^\prime & B_i^\prime
    \end{bmatrix}$ such that
    \begin{equation*}
        \mathcal{P} = \co \Big( \underbrace{\left\lbrace
            \begin{bmatrix}
                A_1^\prime & B_1^\prime
            \end{bmatrix}, 
            \cdots, 
            \begin{bmatrix}
                A^\prime_{N^\prime} & B^\prime_{N^\prime}
            \end{bmatrix} 
        \right\rbrace}_{\eqqcolon \mathcal{S}^\prime } \Big).
    \end{equation*}
    Every convex combination of the matrices $B_i^\prime$ can be written as a convex combination of the original $B_i$ matrices, and every convex combination of the $B_i$'s is full column rank, which implies that every convex combination of the $B_i^\prime$ is also full column rank. The vertices, edges and faces of $\mathcal{P}$ are obtained by intersecting the previous vertices, edges and faces from $\co ( \mathcal{S} )$ with the set $\mathcal{T}$. Since $\begin{bmatrix}
        A_p & B_p
    \end{bmatrix}$ is not on a vertex, edge or face of $\co ( \mathcal{S} )$ it cannot be on a vertex, edge or face of $\mathcal{P}$ and must be in the interior. Hence, the set $\mathcal{S}^\prime$ satisfies~\Cref{pr:assumption_convex_hull}, and also every element of $\mathcal{S}^\prime$ satisfies \eqref{pr:fixedKi} and \eqref{pr:fixedLi}.

\end{proof}

\begin{remark}
    The results provided in~\Cref{propCornerModel} deal with systems that do not have the same number of inputs as states; the special case of $n=m$ is solved trivially, i.e. we have $\mathcal{S} = \mathcal{S}^\prime$. Since every convex combination of $B_i$ is full column rank, that means that $B_i$ can be inverted for all $i \in \left\{1,\cdots,N\right\}$, and the gains can be obtained as
\begin{align*}
    K_i &= B_i^{-1} (A_r - A_i) \\
    L_i &= B_i^{-1} B_r.
\end{align*}
\demo
\end{remark}

\subsection{Obtaining the sets \texorpdfstring{$\mathcal{S}$}{S} and \texorpdfstring{$\mathcal{S}^\prime$}{S'}}
\label{subs:obstainingS}


The modeling uncertainty in~\eqref{pr:plant} is taken to be parametric uncertainty in the matrices in this state-space system model. There is an implicit assumption here that the state-space system model~\eqref{pr:plant} is derived from physical laws rather than from a state-space realization of an input-output system model.
If we consider the maximum range of values each entry of $A_p$, and $B_p$ can take we get that we can bound $a_{ij} \in [a_{ij\min}, a_{ij\max}]$ and $b_{ij} \in [b_{ij\min}, b_{ij\max}]$. When we consider the minimums and maximums of every entry we write the following matrices
\begin{align*}
    A_{p\min} \leq A_p \leq A_{p\max}, \\
    B_{p\min} \leq B_p \leq B_{p\max},
\end{align*}
where $A_{p\min}$ is a matrix where each entry $a_{ij}$ takes its minimum value, and $A_{p\max}$ is a matrix where each entry $a_{ij}$ takes its maximum value, $B_{p\min}$, and $B_{p\max}$ are matrices that take the minimums and maximums of each entry $b_{ij}$, respectively, and the inequality is considered entry-wise. Let $\mathcal{S}$ be the set of all possible system matrices $\begin{bmatrix} A_p & B_p \end{bmatrix}$ each entry of which is either the corresponding entry of 
$[A_{p\min}~~B_{p\min}]$ or the corresponding entry of  $[A_{p\max}~~B_{p\max}]$. Note that, by construction, $\mathcal{S}$ satisfies~\Cref{pr:assumption_convex_hull}~(i). The number of elements in $\mathcal{S}$ is $N = 2^{n(n + m)}$. If the entries can be parameterized by an uncertainty vector $\eta \in \mathbb{R}^{q}$, then $\mathcal{S}$ can be defined in terms of $2^q$ elements (see~\Cref{remarkParameterization}).

The na\"{i}ve corner model selection process described above guarantees that $ \begin{bmatrix} A_p & B_p \end{bmatrix} \in \interior \left( \co \left( S \right) \right) $. However, it is not evident that~\Cref{pr:assumption_convex_hull}~(ii) is satisfied. The results from~\cite{KolodziejczakBarbara1999Ccom} can be used to verify this condition. If~\Cref{pr:assumption_convex_hull}~(ii) is satisfied, then the last step is to use constructive procedure from the proof of~\Cref{propCornerModel} to obtain $\mathcal{S}^\prime$.

\subsection{Example}
\label{subs:modelexample}


In this example we consider a system with two states and one input. We assume that the plant only has uncertainty in the $B_p$ matrix. The \emph{unknown} input matrix to the system is $B_p = \begin{bmatrix}
    2 \\ 2
\end{bmatrix}$. The reference models input's matrix is $B_r = \begin{bmatrix}
    10 \\ 10
\end{bmatrix}$. Note that by taking $L^* = 5$ we can satisfy~\eqref{pr:realLstar}, and~\Cref{pr:assumption_geometric} is satisfied. The polytopic uncertainty for $B_p$ is given by
\begin{equation*}
    B_{p\min} =
    \begin{bmatrix}
        1 \\
        1
    \end{bmatrix},
    B_{p\max} =
    \begin{bmatrix}
        4 \\
        5
    \end{bmatrix}.
\end{equation*}
Using the selection procedure from~\Cref{subs:obstainingS} we can take all possible combinations of the minimums and maximums of every entry of $B_{p\min}$ and $B_{p\max}$ to get
\begin{align*}
    B_1 = 
    \begin{bmatrix}
        1 \\ 1
    \end{bmatrix},
    B_2 =
    \begin{bmatrix}
        1 \\ 5
    \end{bmatrix},
    B_3 = 
    \begin{bmatrix}
        4 \\ 5
    \end{bmatrix},
    B_4 = 
    \begin{bmatrix}
        4 \\ 1
    \end{bmatrix}.
\end{align*}
Since $A_p$ is completely known, we need to redefine the sets $\mathcal{S}$ and $\mathcal{T}$ as
\begin{align*}
    \mathcal{S} &= \left\lbrace
    B_1, 
    \cdots, 
    B_4
    \right\rbrace, \\
    \mathcal{T} &=
        \big\lbrace
            B_x \in \mathbb{R}^{2\times 1}: B_x L = B_r, L \in \mathbb{R}
        \big\rbrace.
\end{align*}
It is easily verifiable that $B_p\in\interior \left( \co \left( \mathcal{S} \right) \right)$, and that any convex combination of $B_i$ is full column rank. Nevertheless, if we consider $B_2$, $B_3$, or $B_4$ we cannot satisfy~\eqref{pr:fixedLi}. Using the proof of~\Cref{propCornerModel} we solve for the set $\mathcal{P}$ graphically (closed line segment going from $B_1$ to $B_2^\prime$), as shown in~\Cref{fig:projectioncorner} to obtain
\begin{align*}
    \mathcal{P} = \co ( \mathcal{S} ) \cap \mathcal{T} = 
    \co \Bigg(  
    \underbrace{\left\{
    \begin{bmatrix}
        1 \\ 1
    \end{bmatrix}, 
    \begin{bmatrix}
        4.5 \\ 4.5
    \end{bmatrix}
    \right\}}_{\mathcal{S}^\prime}
    \Bigg).
\end{align*}

Note that $\begin{bmatrix}
    A_p & B_p
\end{bmatrix} \in \interior \left( \mathcal{P} \right)$. In addition, every convex combination of $B_1$ and $B_2^\prime$ is full column rank, and we can satisfy~\eqref{pr:fixedLi} with $L_1 = 10$, and $L_{2^\prime} = 20/9$.

\begin{figure}
    \centering
    \includegraphics{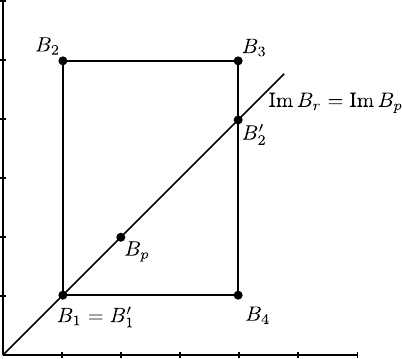}
    \caption{Graphic illustration of the projection process to select corner models.}
    \label{fig:projectioncorner}
\end{figure}
\section{Online Parameter Identification}
\label{sec:sysidentification}

Our multiple-model reference adaptive control (MMRAC) design will utilize a blending-based multiple model parameter identification (MMPI) scheme that generates estimates $\begin{bmatrix}
    \hat{A}_p(t) & \hat{B}_p(t)
\end{bmatrix}$ of $\begin{bmatrix}
    A_p & B_p
\end{bmatrix}$, as a weighted sum of the corner models. In this section, we provide the design, stability and convergence analysis of this MMPI scheme.

\subsection{Multiple-Model Parameter Identifier Design}

Filtering both sides of \eqref{pr:plant} by the linear filter $ \dfrac{1}{s + \lambda} $, where $ \lambda > 0 $ is a design parameter, we obtain the parametric model
\begin{align}
    z \left( t \right) &\stackrel{\lambda}{=} \Theta_p
    \begin{bmatrix}
        \phi_1 \left( t \right) \\
        \phi_2 \left( t \right)
    \end{bmatrix}
    \eqqcolon \Theta_p \Phi (t),
    \label{des:parametric_model}
\end{align}
where $\phi_1 \left( t \right), z \left( t \right) \in \mathbb{R}^{n}$, $\phi_2 \left( t \right) \in \mathbb{R}^{m}$ are generated by the filters
\begin{align}
    \begin{bmatrix}
        \dot{\phi}_1 (t) \\
        \dot{\phi}_2 (t)
    \end{bmatrix}
    &= -\lambda
    \begin{bmatrix}
        \phi_1 (t) \\
        \phi_2 (t)
    \end{bmatrix}
    +
    \begin{bmatrix}
        x_p (t) \\
        u (t)
    \end{bmatrix},\:
    \begin{bmatrix}
        \phi_1 (0) \\
        \phi_2 (0)
    \end{bmatrix}
    =
    \begin{bmatrix}
        0 \\
        0
    \end{bmatrix}, \nonumber \\
    z (t) &= \dot{\phi}_1 (t) = - \lambda \phi_1 (t) + x_p (t).
    \label{des:filtered_plant}
\end{align}
The relation \eqref{des:parametric_model} can be verified by taking time derivatives of both sides and taking the difference, i.e., \eqref{des:filtered_plant} implies for $e_z(t) = z(t) - \Theta_p \Phi(t)$ that
\begin{align}
    \dot{e}_z (t) &= \dot{z} (t) - \Theta_p \dot{\Phi} (t) \nonumber \\
    &= -\lambda \dot{\phi}_1(t) + \dot{x}_p - \Theta_p \begin{bmatrix}
        \dot{\phi}_1(t) \\
        \dot{\phi}_2(t)
    \end{bmatrix} \nonumber \\
    &= -\lambda \dot{\phi}_1(t) + \Theta_p \begin{bmatrix}
        x_p(t) \\
        u(t)
    \end{bmatrix}
    - \Theta_p
    \left(
    -\lambda \begin{bmatrix}
        \phi_1(t) \\
        \phi_2(t)
    \end{bmatrix}
    +
    \begin{bmatrix}
        x_p(t) \\
        u(t)
    \end{bmatrix}
    \right) \nonumber \\
    &= -\lambda\left( -\lambda\phi_1(t) + x_p(t) \right) + \lambda\Theta_p \Phi(t) \nonumber \\
    &= -\lambda\left( -\lambda\phi_1(t) + x_p(t) - \Theta_p \Phi(t) \right) \nonumber \\
    &= -\lambda e_z(t), \label{eq:e_z}
\end{align}
i.e., $e_z(t)$ is an exponentially decaying signal.

For each of the fixed models, define
\begin{align}
    z_i \left( t \right) &\coloneqq \Theta_i 
    \begin{bmatrix}
        \phi_1 \left( t \right) \\
        \phi_2 \left( t \right)
    \end{bmatrix} = \Theta_i \Phi (t),\: \forall i \in \left\lbrace 1, \cdots, N \right\rbrace.
    \label{des:filtered_fixed_models}
\end{align}
Then, the filtered state estimation error for each of the $ N $ fixed models is defined as
\begin{equation}
	\varepsilon_i \left( t \right) \coloneqq \dfrac{z \left( t \right) - z_i \left( t \right)}{m_s^2 (t)}, \forall i \in \left\lbrace 1,\cdots, N \right\rbrace,
	\label{des:filtered_error}
\end{equation}
where $ m_s^2 (t) \coloneqq 1 + \alpha \|\Phi(t)\|^2$, $ \alpha > 0 $, is a normalization signal which guarantees that $ \dfrac{\Phi (t)}{m_s (t)} $ is bounded. By \eqref{pr:convexity}, \eqref{des:parametric_model}--\eqref{des:filtered_fixed_models}, for any $w=\left[ w_1,\dots,w_{N} \right]^\top  \in \mathcal{W}$, we obtain
\begin{align}
    z \left( t \right) &= \sum_{i=1}^{N} w_i z_i \left( t \right)+e_z(t).
    \label{des:z_zi}
\end{align}
Using $ \sum_{i=1}^{N} w_i = 1 $, \eqref{des:z_zi} further implies
\begin{align}
    e_z(t)= \sum_{i=1}^{N} w_i z \left( t \right) - \sum_{i=1}^{N} w_i z_i \left( t \right)
    &= \sum_{i=1}^{N} w_i \left( z \left( t \right) - z_i \left( t \right) \right) \nonumber \\
    &= \sum_{i=1}^{N} w_i \varepsilon_i \left( t \right) m_s^2 (t). \label{des:zeroerror}
\end{align}
Adding $-\varepsilon_{N} \left( t \right) m_s^2 (t)$ to both sides of \eqref{des:zeroerror}, we obtain
\begin{align*}
    e_z(t)-\varepsilon_{N} \left( t \right) m_s^2 (t) 
    &= \sum_{i=1}^{N-1} w_i \left( \varepsilon_i \left( t \right) - \varepsilon_{N} \left( t \right) \right) m_s^2 (t),
\end{align*}
which implies, for any $w=\left[ w_1,\dots,w_{N} \right]^\top  \in \mathcal{W}$, that
\begin{equation}
    e_z(t)-\varepsilon_{N} \left( t \right)  = \sum_{i=1}^{N - 1} w_i \left( \varepsilon_i \left( t \right) - \varepsilon_{N} \left( t \right) \right).
    \label{des:convex_error}
\end{equation}
Defining the $n \times (N-1)$ time-varying matrix
\begin{equation}
    E \left( t \right) \coloneqq \begin{bmatrix}
	\varepsilon_1 \left( t \right) - \varepsilon_{N} \left( t \right) & \cdots & \varepsilon_{N-1}  \left( t \right) - \varepsilon_{N} \left( t \right)
    \end{bmatrix}, \label{des:defineE}
\end{equation}
we can rewrite \eqref{des:convex_error} in matrix form as
\begin{equation}
	E \left( t \right) \bar{w}  = e_z(t) -\varepsilon_{N} \left( t \right)
	\label{des:convex_error_matrix}
\end{equation}
for any $\left [\bar{w}^\top,w_N\right]^\top=\left[ w_1,\dots,w_{N-1},w_{N} \right]^\top  \in \mathcal{W}$.
\Cref{des:convex_error_matrix} motivates using the following recursive adaptive law \cite{ioannou_adaptive_2006} to generate estimate $\hat{w} (t) = [\hat{\bar{w}}^\top \left( t \right), \hat{w}_{N} \left( t \right)]^\top \in \mathbb{R}^{N}$, 
such that $\lim_{t\rightarrow\infty} \sum_{i = 1}^{N} \hat{w}_i (t)
            \begin{bmatrix}
                A_i & B_i
            \end{bmatrix} = \begin{bmatrix}
                A_p & B_p
            \end{bmatrix}$.
\begin{align}
    \begin{split}
        \dot{\hat{\bar{w}}} \left( t \right) &= -\Gamma E^\top \left( t \right) E \left( t \right) \hat{\bar{w}} \left( t \right) - \Gamma E^\top \left( t \right) \varepsilon_{N} \left( t \right), \\
	    \hat{w}_{N} \left( t \right) &= 1 - \sum_{i=1}^{N-1} \hat{w}_i \left( t \right),
    \end{split}
	\label{des:gradient_algorithm}
\end{align}
where the tuning parameter $\Gamma \in \mathbb{R}^{(N-1) \times (N-1)} $ is a symmetric positive definite matrix. Let $w^*$ be some element in $\mathcal{W}$, and based on \eqref{des:gradient_algorithm}, the estimation error $ \tilde{\bar{w}} \left( t \right) \coloneqq \hat{\bar{w}} \left( t \right) - \bar{w}^*$ satisfies
\begin{align}
    \dot{\tilde{\bar{w}}} \left( t \right) &= \dot{\hat{\bar{w}}} \left( t \right) - \dot{\bar{w}}^* = \dot{\hat{\bar{w}}} \left( t \right) \nonumber \\
    &= - \Gamma E^\top \left( t \right) E \left( t \right) \hat{\bar{w}} \left( t \right) - \Gamma E^\top \left( t \right) \varepsilon_{N} \left( t \right). \label{des:estimation_error_hat}
\end{align}
Pre-multiplying \eqref{des:convex_error_matrix} by $\Gamma E^\top(t)$ and moving all the terms to the left hand side, we obtain
\begin{equation}
	\Gamma E^\top(t)E \left( t \right) \bar{w}+\Gamma E^\top(t)\varepsilon_{N} \left( t \right)-\Gamma E^\top(t) e_z(t)  =  0
	\label{des:convex_error_matrix2}
\end{equation}
Adding \eqref{des:convex_error_matrix2} to \eqref{des:estimation_error_hat}, we further obtain
\begin{align}
    \dot{\tilde{\bar{w}}} \left( t \right) & = -\Gamma E^\top \left( t \right) E \left( t \right) \tilde{\bar{w}} \left( t \right)-\Gamma E^\top(t) e_z(t). \label{des:estimation_error}
\end{align} 

\subsection{Stability and Convergence of the Identification Process}

In this subsection, we establish the stability and convergence properties of the estimation scheme \eqref{des:gradient_algorithm} utilizing the following lemma.

\begin{lemma}[{See~\cite{ioannou_adaptive_2006}}, Barbalat's Lemma]
	For a function $f:[0,\infty)\rightarrow \mathbb{R}$, if $ f,\: \dot{f} \in \mathcal{L}_\infty $ and $ {f \in \mathcal{L}_p} $, for some $ p \in \left[ 1, \infty \right) $, then $ \lim\limits_{t \rightarrow \infty} f \left( t \right) = 0 $. 
	\label{int:lemma}
\end{lemma}

The main stability and convergence properties are established in the following theorems and lemmas below.
\begin{theorem}
	Consider the system \eqref{pr:plant} with definitions \eqref{des:filtered_error}, and \eqref{des:defineE}. Let \Cref{pr:assumption_convex_hull} hold, $w^*=\left [\bar{w}^{*\top},w^*_N\right]^\top \in \mathcal{W}$ be an arbitrary vector within the set defined in \eqref{pr:convexity}, and $ \tilde{\bar{w}}(t) \coloneqq \hat{\bar{w}}(t) - \bar{w}^*$.  For any initial condition $ \hat{\bar{w}} (0) \in \mathbb{R}^{N-1} $, the estimation scheme \eqref{des:gradient_algorithm} guarantees that:
	\begin{enumerate}[label=(\roman*)]
		\item $\hat{\bar{w}}$, $\tilde{\bar{w}}$, $\dot{\tilde{\bar{w}}}$ and $ E $ are bounded signals. \label{des:theorem1i}
		\item $ E \tilde{\bar{w}} $ and $\dot{\tilde{\bar{w}}}$ are square integrable.  \label{des:theorem1ii} 
        \item $ \lim_{t\rightarrow\infty} E(t)\tilde{\bar{w}}(t) = 0 $. \label{des:theorem1iii}
 \item $\hat{\bar{w}}(t)$ asymptotically converges to a constant vector $\bar{\bar{w}}\in\mathbb{R}^{N-1}$. \label{des:theorem1iv}
	\end{enumerate}
	\label{des:theorem1}
\end{theorem}
\begin{proof}
    Consider the Lyapunov-like function
    \begin{equation}
        V_1(\tilde{\bar{w}}(t), e_z(t)) = \frac{1}{2}\tilde{\bar{w}}^\top (t)\Gamma^{-1} \tilde{\bar{w}} (t) + \frac{1}{2\lambda}e_z^\top(t)e_z(t). 
        \label{des:V1}
    \end{equation}
    Taking the time derivative of~\eqref{des:V1} along~\eqref{des:estimation_error}, we have
 \begin{align}
            \dfrac{\D V_1(t)}{\D t} &= \tilde{\bar{w}}^\top(t) \Gamma^{-1} \dot{\tilde{\bar{w}}}(t) + \dfrac{1}{\lambda} e_z^\top(t) \dot{e}_z(t) \nonumber \\
            &= -\tilde{\bar{w}}^\top(t)\Gamma^{-1}\Gamma E^\top(t) E(t) \tilde{\bar{w}}(t) \nonumber \\
            &- 
            \tilde{\bar{w}}^\top(t) \Gamma^{-1}\Gamma E^\top(t) e_z(t) + \dfrac{1}{\lambda} e_z^\top\dot{e}_z \nonumber \\
            &= -\tilde{\bar{w}}^\top(t) E^\top(t) E(t) \tilde{\bar{w}}(t) - 
            \tilde{\bar{w}}^\top(t) E^\top(t) e_z(t) + \dfrac{1}{\lambda} e_z^\top\dot{e}_z \nonumber \\
            &= - \dfrac{1}{2}\|E \tilde{\bar{w}}\|^2 + \dfrac{1}{2} \| e_z(t) \|^2 - \dfrac{1}{2} \| E(t)\tilde{\bar{w}}(t) + e_z(t) \|^2.
            \label{des:V1dot}
        \end{align}
    Since $e_z(t)$ is an exponentially decaying signal, \Cref{des:V1dot} implies that $\dot{V}_1$ is bounded and there exists a time instant $t_1>0$ such that, for all $t \geq t_1$, 
    $\dot{V}_1(t) \leq 0$. Hence.  $V_1$, $\tilde{\bar{w}}$, and $\hat{\bar{w}}$ are bounded. Because of normalization~\eqref{des:filtered_error}, $\varepsilon_i$ terms are bounded and hence $E$ is bounded, which together with \eqref{des:estimation_error} also implies that $\dot{\tilde{\bar{w}}}$ is bounded,  finishing the proof of (i).
    
    Since $V_1(t)$ is always positive, bounded, and decaying, the integral of \eqref{des:V1dot} for $t=0$ to $\infty$ is finite, and hence  $ E \tilde{\bar{w}} $ and $E \tilde{\bar{w}}+e_z(t)$ are square integrable.
    Since $E$ is bounded, this, together with \eqref{des:estimation_error}, further implies that $\dot{\tilde{\bar{w}}}$ is square integrable, completing the proof of (ii).

    Items (i) and (ii) together with Barbalat's Lemma imply (iii). Items (i) and (ii) further imply that $\lim_{t\rightarrow\infty}\tilde{\bar{w}}(t)$ and, hence, $\lim_{t\rightarrow\infty}\hat{\bar{w}}(t)$ exist and are finite, proving (iv).
\end{proof}

\Cref{des:theorem1} implies that $ \| E(t)\tilde{\bar{w}}(t) \| $ asymptotically converges to zero, but this does not mean that $ \tilde{\bar{w}}(t) $ converges to the set $\mathcal{W}$. We can now state, and prove, the main result of the identification process.

\begin{theorem}
Consider the system \eqref{pr:plant} with definitions \eqref{des:filtered_error}, and \eqref{des:defineE}, and the estimation scheme \eqref{des:gradient_algorithm}. Let \Cref{pr:assumption_convex_hull} hold, $w^*=\left [\bar{w}^{*\top},w^*_N\right]^\top \in \mathcal{W}$ be an arbitrary vector within the set defined in \eqref{pr:convexity}, $ \tilde{\bar{w}}(t) \coloneqq \hat{\bar{w}}(t) - \bar{w}^*$, and  $\Phi$ be bounded and satisfy the PE condition~\eqref{eq:PEinequality}. Then, for any initial condition $ \hat{\bar{w}} (0) \in \mathbb{R}^{N-1} $, the estimated system matrix $ \sum_{i=1}^{N} \hat{w}_i(t) \Theta_i $ asymptotically converges to $ \Theta_p $.
	\label{des:theorem2}
\end{theorem}

\begin{proof}
    Let $Q(t) \coloneqq \sum_{i = 1}^{N-1} (\tilde{w}_i(t)(\Theta_N-\Theta_i)) \in \mathbb{R}^{n\times (n+m)}$. It is easily verifiable that if $Q(t) = 0$, then $ \sum_{i=1}^{N} \hat{w}_i(t) \Theta_i = \Theta_p $, and $\hat{w}(t)\in\mathcal{W}$. Theorem \ref{des:theorem1} (iv) implies that $Q(t)$ asymptotically converges to  a constant matrix $\bar{Q}\in \mathbb{R}^{n\times (n+m)}$.
    Hence, to establish that
    $ \sum_{i=1}^{N} \hat{w}_i(t) \Theta_i $ asymptotically converges to $ \Theta_p $,
    we will show that $\bar{Q}=0$.
    
    Expressing~\Cref{des:theorem1}~\ref{des:theorem1iii} in summation form, we get
    \begin{align}
    \label{eq:Thm2proof}
        \lim_{t\rightarrow \infty} E(t) \tilde{\bar{w}}(t) &= \lim_{t\rightarrow \infty}  Q(t) \dfrac{\Phi(t)}{1 + \alpha \|\Phi(t)\|^2}
        \nonumber \\
        &= \lim_{t\rightarrow \infty} \bar{Q} \dfrac{\Phi(t)}{1 + \alpha \|\Phi(t)\|^2}
        =0.
    \end{align}
    Since $\Phi$ is assumed to be bounded, \eqref{eq:Thm2proof} implies that
     \begin{align*}
    \label{eq:Thm2proof2}
    \lim_{t\rightarrow \infty} \bar{Q} \Phi(t)=0.
    \end{align*}
    and hence
    \begin{align}
    \lim_{t\rightarrow \infty} \bar{Q} \Phi(t)\Phi^\top(t)\bar{Q}^\top=0.
    \end{align}
    Since $\Phi$ satisfies the PE condition~\eqref{eq:PEinequality}, Equation \eqref{eq:Thm2proof2} implies that $\bar{Q}=0$, completing the proof.  
\end{proof}

The recursive adaptive law \eqref{des:gradient_algorithm} guarantees that $ \sum_{i=1}^{N} \hat{w}_i \left( t \right) = 1 $, but not that $ \hat{w} \left( t \right) \in [0,1]^N $. Since the set $ \left[ 0,1 \right]^N $ is convex, the projection of $ \hat{w} \left( t \right) $ into $ \left[ 0,1 \right]^N $ is well-defined.

Define the compact set
\begin{equation}
    \Pi \coloneqq \left\lbrace \hat{\bar{w}} \in \left[ 0,1 \right]^{N-1} : \sum_{i=1}^{N-1}\hat{\bar{w}}_i \leq 1 \right\rbrace, \label{des:compactset}
\end{equation}
and let $ \proj_{\Pi,\hat{\bar{w}}} : \mathbb{R}^{N-1} \rightarrow \Pi \subset \mathbb{R}^{N-1} $ denote the parameter projection operator~\cite{ioannou_adaptive_2006}. Choose $ \hat{w}(0) \in \interior(\Pi)$, then the recursive adaptive algorithm \eqref{des:gradient_algorithm} with parameter projection is as follows:
\begin{align}
    \dot{\hat{\bar{w}}} \left( t \right) &= \proj_{\Pi,\hat{\bar{w}}} \left( -\Gamma \left( E^\top \left( t \right) E \left( t \right) \hat{\bar{w}} \left( t \right) + E^\top \left( t \right) \varepsilon_{N} \left( t \right) \right) \right), \nonumber \\
    \hat{w}_{N} \left( t \right) &= 1 - \sum_{i=1}^{N-1} \hat{w}_i \left( t \right). \label{des:gradient_algorithm_projection}
\end{align}
The parameter projection operator enforces that $\Pi$ is a positively invariant subset for the dynamics~\eqref{des:gradient_algorithm_projection}.
\begin{corollary}
	The gradient adaptive law with parameter projection \eqref{des:gradient_algorithm_projection}  has all the properties established in \Cref{des:theorem1} for the gradient adaptive law \eqref{des:gradient_algorithm}. Furthermore, we get the following properties:
    \begin{enumerate}[label=(\roman*)]
        \item If $ \hat{w} \left( 0 \right) \in \Pi $, then $ \hat{w} \left( t \right) \in \Pi $, $ \forall t \geq 0 $.
        \item If $ \hat{w} \left( 0 \right) \in \interior \Pi $, then $ \hat{w} (t) $ asymptotically converges to the set $ \mathcal{W} $.
    \end{enumerate}
    \label{des:corollary1}
\end{corollary}
\begin{proof}
	The proof follows applying Theorem 3.10.1 of \cite{ioannou_adaptive_2006} to \eqref{des:gradient_algorithm} combined with \Cref{des:theorem1,des:theorem2}.
\end{proof}

\section{Multiple Model Reference Adaptive Control}
\label{sec:syscontrol}

In this section we combine the system identification scheme of \Cref{sec:sysidentification} with a MMRAC controller to achieve asymptotic tracking of the reference system \eqref{pr:model_reference}.

\subsection{Multiple Model Reference Adaptive Control Design}

We use the recursive adaptive algorithm with parameter projection \eqref{des:gradient_algorithm_projection} to design a MMRAC scheme to achieve the adaptive state tracking control task stated in \Cref{sec:Problem_Formulation}.  We utilize \Cref{des:theorem1}~\ref{des:theorem1iii} to construct the proposed MMRAC scheme.

If $ w \in \mathcal{W} $, then multiplying both sides of~\eqref{pr:fixedLi} by $w_i$ and summing over $ i $ yields
\begin{align*}
    \sum_{i=1}^{N} w_i B_i L_i &= \sum_{i=1}^{N} w_i B_r = B_r,
\end{align*}
which implies, together with \eqref{pr:realLstar} from \Cref{pr:assumption_geometric}, that
\begin{equation}
	B_p L^{*} = \sum_{i=1}^{N} w_i B_i L_i.
	\label{des:Lstar}
\end{equation}
Applying the same steps on \eqref{pr:fixedKi} we get
\begin{align}
    \begin{split}
        A_r = \sum_{i=1}^{N} w_i A_r &= \sum_{i=1}^{N} w_i \big(A_i + B_i K_i\big) \\
        &= A_p + \sum_{i=1}^{N} w_i B_i K_i.
    \end{split}
    \label{des:sumKi}
\end{align}
Comparing \eqref{des:sumKi} to \eqref{pr:realKstar} we get that
\begin{equation}
	B_p K^{*} = \sum_{i=1}^{N} w_i B_i K_i.
	\label{des:Kstar}
\end{equation}
Equations~\eqref{des:Lstar} and~\eqref{des:Kstar} motive us to generate estimates of the gains $ K^{*} $ and $ L^{*} $ using the estimates $ \hat{w} \left( t \right)$, 
keeping in mind the rank supposition in Assumption~\ref{pr:assumption_convex_hull}, as
\begin{align}
        \hat{K} \left( t \right) &= \hat{B}_p^{\dagger} (t) \sum_{i=1}^{N} \hat{w}_i \left( t \right) B_i K_i, \label{des:Khat} \\
        \hat{L} \left( t \right) &= \hat{B}_p^{\dagger} (t) \sum_{i=1}^{N} \hat{w}_i \left( t \right) B_i L_i, \label{des:Lhat}
\end{align}
where $ \hat{B}_p^{\dagger} (t) $ is the Moore-Penrose inverse of 
\begin{equation}
    \hat{B}_p (t) =\sum_{i=1}^{N} \hat{w}_i \left( t \right) B_i.
    \label{des:Bphat}
\end{equation}
\Cref{des:Kstar,des:Lstar,des:Khat,des:Lhat}, together with \Cref{des:corollary1} further imply the following:
\begin{corollary}
    Consider the system \eqref{pr:plant} with definitions \eqref{des:filtered_error}, \eqref{des:defineE}, and the reference model \eqref{pr:model_reference}. If $\Phi(t)$ satisfies the PE condition with constants $\alpha_{\Phi 1}$, $\alpha_{\Phi 2}$, and $T_{\Phi}$, then the estimates $ \hat{K} \left( t \right) $ and $ \hat{L} \left( t \right) $  defined in  \eqref{des:Khat}, and \eqref{des:Lhat} are bounded, and asymptotically converge to $ K^{*} $ and $ L^{*} $, respectively.
	\label{des:corollary2}
\end{corollary}
\begin{proof}
    Since $ \Pi $ is compact, we get that $ \hat{B}_p (t) $ belongs to a compact set, which means the Moore-Penrose inverse of $ \hat{B}_p (t) $ exists, and it is bounded. Furthermore, from \Cref{pr:assumption_convex_hull}, we have that $ \hat{B}_p (t) $ is full column rank, which means that we can write
    \begin{equation}
        \hat{B}_p^{\dagger} (t) = \left( \hat{B}_p^\top (t) \hat{B}_p (t)  \right)^{-1} \hat{B}_p^\top (t).
        \label{des:Bphat_inv}
    \end{equation}    
    Combining \eqref{des:Bphat_inv} with \Cref{des:Khat,des:Lhat} we get that $ \hat{K} \left( t \right) $, and $ \hat{L} \left( t \right) $ are bounded. The proof to show asymptotic convergence is the same as the proof for \Cref{des:theorem1}~\ref{des:theorem1iii}.
\end{proof}
The control law we will consider to achieve asymptotic tracking of the reference model is
\begin{equation}
    u \left( t \right) \coloneqq \hat{K} \left( t \right) x_p \left( t \right) + \hat{L} \left( t \right) r \left( t \right).
    \label{des:controllaw}
\end{equation}

\subsection{Stability Analysis of the Entire System}

The main result of the paper is now presented.

\begin{theorem}
    Consider the plant \labelcref{pr:plant} and the reference model \labelcref{pr:model_reference}. If $\Phi(t)$ satisfies the PE condition \eqref{eq:PEinequality}, and \Cref{pr:assumption_convex_hull,pr:assumption_geometric} hold, then the MMRAC scheme \labelcref{des:filtered_plant,des:filtered_error,des:filtered_fixed_models,des:gradient_algorithm_projection,des:Khat,des:Lhat,des:Bphat,des:Bphat_inv,des:controllaw} guarantees that for any
    \begin{enumerate}[label=(\roman*)]
        \item initial conditions of the plant \eqref{pr:plant},
        \item initial conditions of the reference model \eqref{pr:model_reference}, and
        \item piecewise continuous and bounded reference signal $ {r : \left[ 0,\infty \right) \rightarrow \mathbb{R}^{m}} $ in \eqref{pr:model_reference},
    \end{enumerate}
    all closed-loop signals are bounded and $ x_p \left( t \right) $ asymptotically converges to $ x_r \left( t \right) $.
\end{theorem}

\begin{proof}
    Let $x_p(0)=x_{p0}$ and $x_r(0)=x_{r0} $ be arbitrary initial plant and reference model states, and $ r(t) $ be any known, bounded, piecewise continuous reference signal. Substituting \eqref{des:controllaw} into \eqref{pr:plant} we get
    \begin{align*}
        \dot{x}_p \left( t \right) &= A_p x_p \left( t \right) + B_p u \left( t \right) \\
        &= A_p x_p \left( t \right) + B_p \hat{K} \left( t \right) x_p \left( t \right) + B_p \hat{L} \left( t \right) r \left( t \right).
    \end{align*}
    Adding and subtracting $ B_p K^{*} x_p \left( t \right) $ and $ B_p L^{*} r \left( t \right) $, defining $ \tilde{K} \left( t \right) \coloneqq \hat{K} \left( t \right) - K^{*} $, $ \tilde{L} \left( t \right) \coloneqq \hat{L} \left( t \right) - L^{*} $, and using \Cref{pr:assumption_geometric} we get
    \begin{align*}
        \dot{x}_p \left( t \right) = A_r x_p \left( t \right) + B_r r \left( t \right) + B_p \tilde{K} \left( t \right) x_p \left( t \right) + B_p \tilde{L} \left( t \right) r \left( t \right).
    \end{align*}
    For the tracking error $ e \left( t \right) \coloneqq x_p \left( t \right) - x_r \left( t \right) $, this implies\footnotesize
    \begin{align}
        \dot{e} \left( t \right) &= A_r e\left( t \right) + B_p \tilde{K} \left( t \right) x_p \left( t \right) + B_p \tilde{L} \left( t \right) r \left( t \right) \nonumber \\
        &= \left( A_r + B_p \tilde{K}(t) \right) e (t) + B_p \tilde{K} (t) x_r (t) + B_p \tilde{L} (t) r (t). \label{des:trackingerror}
    \end{align}
    \normalsize Let $ Q \in \mathbb{R}^{n \times n} $ be a fixed, symmetric, and positive definite matrix. Then we can define $ P \in \mathbb{R}^{n \times n} $ to be the unique positive definite and symmetric solution of $ {PA_r + A_r^\top P + Q = 0} $. Consider 
    the positive definite function
    \begin{align}
    	V \left( e \left( t \right) \right) &= e^\top \left( t \right) P e \left( t \right).
        \label{des:lyapunovcontrol}
    \end{align}
	Taking the derivative of \eqref{des:lyapunovcontrol} along solutions of~\eqref{des:trackingerror} we get
	\begin{align}
	    \dot{V} \left( e \left( t \right) \right) 
     \nonumber \\
	    &=  
     2 e^\top \left( t \right) P \dot{e} \left( t \right) \nonumber \\
  &= 2 e^\top \left( t \right) P A_r e \left( t \right) + 2 e^\top \left( t \right) P B_p \tilde{K} \left( t \right) x_p \left( t \right) \nonumber \\
	    &\qquad + 2 e^\top \left( t \right) P B_p \tilde{L} \left( t \right) r \left( t \right).
     \label{eq:Vdot}
	\end{align}
	The first term satisfies
	\begin{align*}
	    2 e^\top \left( t \right) P A_r e \left( t \right) &= e^\top \left( t \right) P A_r e \left( t \right) + e^\top \left( t \right) P A_r e \left( t \right) \\
	    &= e^\top \left( t \right) \left( P A_r + A_r^\top P \right) e \left( t \right) \\
	    &= - e^\top \left( t \right) Q e \left( t \right).
	\end{align*}
	Hence, \eqref{eq:Vdot} can be rewritten as
 as
	\begin{align*}
	    \dot{V} \left( e \left( t \right) \right) = &- e^\top \left( t \right) Q e \left( t \right) + 2 e^\top \left( t \right) P B_p \tilde{K} \left( t \right) x_p \left( t \right) \\
	    &\qquad+ 2 e^\top \left( t \right) P B_p \tilde{L} \left( t \right) r \left( t \right).
	\end{align*}
	Substituting $ x_p (t) = e (t) + x_r (t) $, we get
	\begin{align*}
	    \dot{V} \left( e \left( t \right) \right) = &- e^\top \left( t \right) \left( Q - P B_p \tilde{K} \left( t \right) \right) e \left( t \right) \\
	    &+ 2 e^\top \left( t \right) P B_p \tilde{K} \left( t \right) x_r \left( t \right) \\
	    &+ 2 e^\top \left( t \right) P B_p \tilde{L} \left( t \right) r \left( t \right).
	\end{align*}
	Defining
    \begin{equation*}
        c_1 (t) \coloneqq \| P B_p \tilde{K} (t) \|,
    \end{equation*}
    and
    \begin{equation*}
        c_2 (t) \coloneqq 2 \| P B_p \tilde{K} (t) \| \| x_r (t) \| + 2 \| P B_p \tilde{L} (t) \| \| r (t) \|,
    \end{equation*}
	we have
	\begin{equation}
	    \dot{V} (e(t)) \leq - \left( \lambda_{\textrm{min}} (Q) - c_1 (t) \right) \| e (t) \|^2 + c_2 (t) \| e (t) \|, \label{des:lyaptheorem2}
	\end{equation}
	where $ \lambda_{\textrm{min}} (Q) > 0 $ is the minimum eigenvalue of $ Q $. Note that \eqref{des:trackingerror} does not have a finite escape time; from \Cref{des:corollary2} we have that $ \hat{K} (t) $ and $ \hat{L} (t) $ are continuous and bounded, which means that $ \tilde{K} (t) $ and $ \tilde{L} (t) $ are also continuous and bounded, hence \eqref{des:trackingerror} may only go to infinity as time goes to infinity. 
    Moreover, if $\Phi(t)$ is PE, then $\tilde{K} (t)$ and $\tilde{L} (t)$ converge to zero asymptotically, which lets us conclude that $c_1 (t)$ is bounded, and there exists $ t_1 \geq 0 $ such that $\lambda_{\textrm{min}} (Q) > c_1 (t)$ for all $ t \geq t_1 $. We can define
    \begin{align}
        \bar{c}_1 = \sup \{ c_1(t), t \geq t_1 \}, \\
        \bar{c}_2 = \sup \{ c_2(t), t \geq t_1 \}.
    \end{align}
    
    This implies that for all $ t \geq t_1 $ we get that if
    \begin{align*}
        \| e(t) \| > \dfrac{\bar{c}_2}{\lambda_{\min} Q - \bar{c}_1},
    \end{align*}
    then $\dot{V} < 0$, and we can conclude that $e$ is bounded.
    This further implies that $x_p$ is bounded, which finally implies that $ u $ is bounded, showing that all signals in the closed-loop system are bounded. Combining this with \Cref{des:corollary2} we get
	\begin{align*}
	    \lim_{t \rightarrow \infty} \left( \dot{V} (e(t)) + e^\top \left( t \right) Q e \left( t \right) \right) = 0.
	\end{align*}
	This implies that $ e (t) $ converges to $ 0 $, asymptotically, i.e., $ x_p (t) $ asymptotically converges to $ x_r (t) $.
\end{proof}
\section{Simulations}
\label{sec:Simulation}

\begin{figure}
    \centering
    \includegraphics[width=1\linewidth]{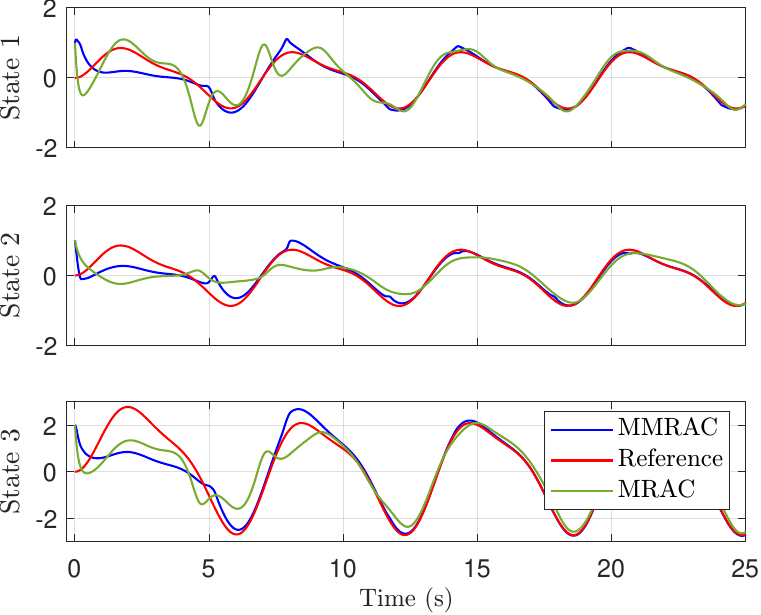}
    \caption{State of the system $ x_p(t) $, and state of the reference model $ x_r(t) $.}
    \label{fig:track}
\end{figure}

\begin{figure}
    \centering
    \includegraphics[width=1\linewidth]{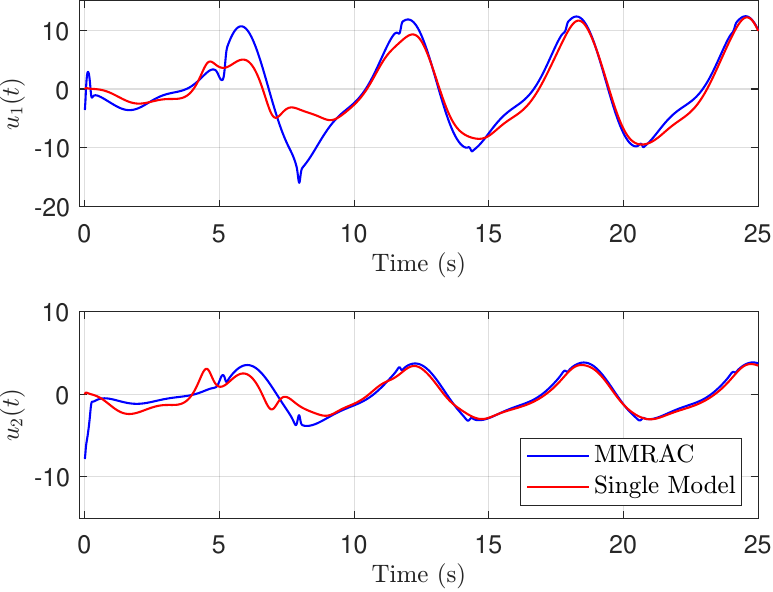}
    \caption{Control efforts for MMRAC and a single model MRAC.}
    \label{fig:controleffort}
\end{figure}

\begin{figure}
    \centering
    \includegraphics[width=1\linewidth]{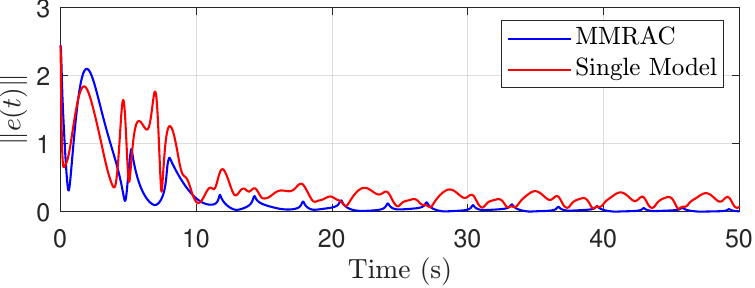}
    \caption{Euclidean norm of the tracking error for MMRAC and a single model MRAC.}
    \label{fig:errornorm}
\end{figure}

\begin{figure}
    \centering
    \includegraphics[width=1\linewidth]{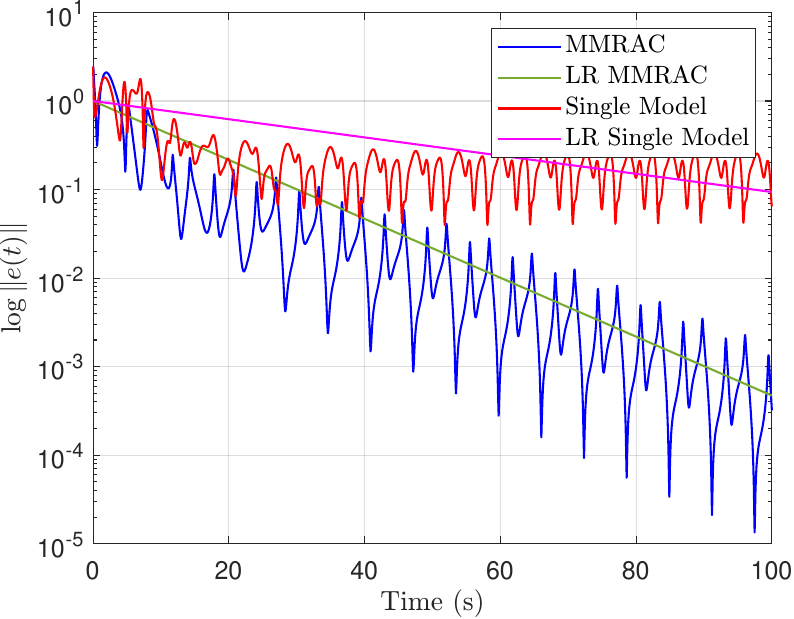}
    \caption{Semilog plot for the norm 2 of the tracking errors, and the linear regressions for MMRAC and a single model.}
    \label{fig:convergencespeed}
\end{figure}

\begin{figure}
    \centering
    \includegraphics[width=1\linewidth]{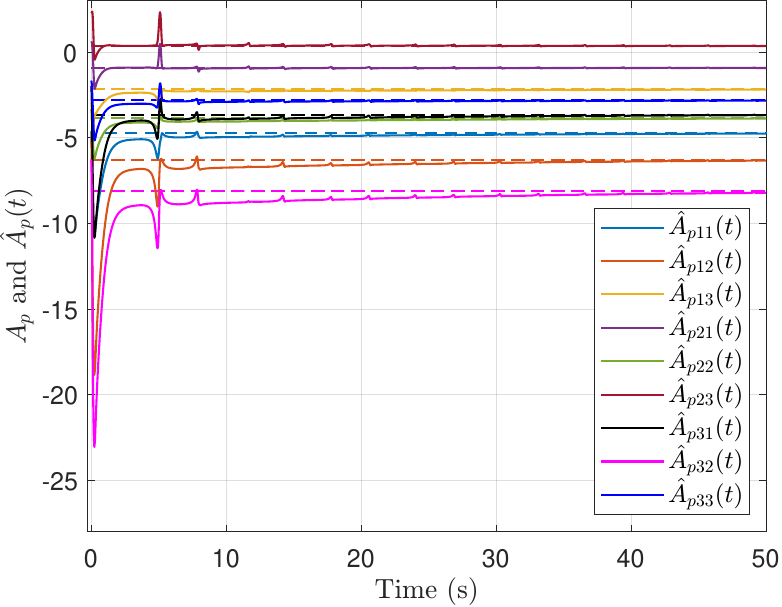}
    \caption{Estimation of each entry of $A_p$. The dashed line is the true value.}
    \label{fig:ApSimulation}
\end{figure}

\begin{figure}
    \centering
    \includegraphics[width=1\linewidth]{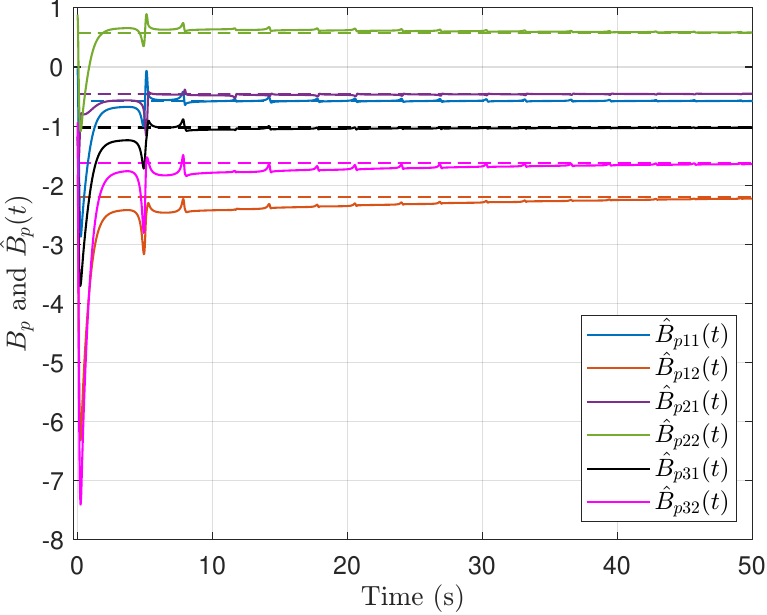}
    \caption{Estimation of each entry of $B_p$. The dashed line is the true value.}
    \label{fig:BpSimulation}
\end{figure}


In this section, we illustrate the behavior and performance of the proposed MMRAC scheme through a set of simulation tests performed on an uncertain model~\eqref{pr:plant}, a reference model~\eqref{pr:model_reference}, and a set $\mathcal{S}$~\eqref{pr:setS} that defines the polytopic uncertainty. We compare the results with simulations using a single model for MRAC.

Consider the uncertain system with $ n = 3 $, $ m = 2 $ given by the matrices
\small
\begin{equation*}
    A_p =
    \begin{bmatrix}
        -4.725 & -6.275 & -2.175 \\
        -0.925 & -3.85 & 0.35 \\
        -3.65 & -8.125 & -2.825
    \end{bmatrix},\:
    B_p =
    \begin{bmatrix}
        -0.575 & -2.2 \\
        -0.45 & 0.575 \\
        -1.025 & -1.625
    \end{bmatrix}.
\end{equation*}
\normalsize
The reference model~\eqref{pr:model_reference} is \small
\begin{equation*}
A_r = \begin{bmatrix}
        -1 & 0 & 0 \\
        0 & -1 & 0 \\
        1 & 1 & -1
    \end{bmatrix},\:
B_r = \begin{bmatrix}
        1 & 0 \\
        0 & 1 \\
        1 & 1
    \end{bmatrix}.
\end{equation*}
\normalsize 
It can be verified that~\Cref{pr:assumption_geometric} is satisfied by defining the matrices
\begin{align*}
    K^* =
    \begin{bmatrix}
        -3.16 & -7.48 & -0.36 \\
        -0.87 & -0.89 & -0.89
    \end{bmatrix},
    L^* =
    \begin{bmatrix}
        -0.44 & -1.66 \\
        -0.34 & 0.44
    \end{bmatrix}.
\end{align*}
Considering the following $5$ matrix pairs
\begin{align*}
    A_1 = \begin{bmatrix}
        -0.75 & 0.25 & 0.25 \\
        -3.75 & -4.75 & -3.75 \\
        -2.5 & -2.5 & -4.5
    \end{bmatrix},
    B_1 = \begin{bmatrix}
        0.25 & -0.5 \\
        1.25 & 2.5 \\
        1.5 & 2
    \end{bmatrix},
\end{align*}
\begin{align*}
    A_2 = \begin{bmatrix}
        -16 & -30 & -5 \\
        -3.5 & -8.5 & -1.5 \\
        -17.5 & -36.5 & -7.5
    \end{bmatrix},
    B_2 = \begin{bmatrix}
        -5 & -10 \\
        -1 & -2.5 \\
        -6 & -12.5
    \end{bmatrix},
\end{align*}
\begin{align*}
    A_3 = \begin{bmatrix}
        2 & 0 & -1 \\
        -0.5 & -1.5 & -0.5 \\
        -0.5 & 0.5 & -2.5
    \end{bmatrix},
    B_3 = \begin{bmatrix}
        -1 & 1 \\
        -0.5 & 0 \\
        -1.5 & 1
    \end{bmatrix},
\end{align*}
\begin{align*}
    A_4 = \begin{bmatrix}
        -1.5 & -0.75 & -0.75 \\
        -0.25 & -0.875 & 1.125 \\
        0.25 & 0.375 & -0.625
    \end{bmatrix},
    B_4 = \begin{bmatrix}
        0.25 & 0.25 \\
        0.125 & -0.375 \\
        0.375 & -0.125
    \end{bmatrix},
\end{align*}
\begin{align*}
    A_5 = \begin{bmatrix}
        -4 & -1 & -5 \\
        5 & -2 & 8 \\
        3 & -1 & 2
    \end{bmatrix},
    B_5 = \begin{bmatrix}
        2 & -1 \\
        -3 & 2 \\
        -1 & 1
    \end{bmatrix},
\end{align*}
we can define the set
\begin{align*}
    \mathcal{S} =
    \left\lbrace
        \begin{bmatrix}
            A_1 & B_1
        \end{bmatrix},
        \cdots,
        \begin{bmatrix}
            A_5 & B_5
        \end{bmatrix}
    \right\rbrace.
\end{align*}
The set $\mathcal{S}$ satisfies~\Cref{pr:assumption_convex_hull}. The design parameters for the estimation scheme and the controller are $ \Gamma = 2 I_{5 \times 5} $, $ \lambda = 0.5 $, and $ \alpha = 0.01 $, with the initial condition $\hat{w}(0) = \begin{bmatrix}
        0.2 & 0.15 & 0.15 & 0.1 & 0.4
    \end{bmatrix}^\top$. The input to the reference model is taken to be $ r_1 \left( t \right) = r_2 \left( t \right) = \sin{ \left( t \right) } + 0.5 \sin{2t}$. With these definitions the full controller is
\begin{align*}
    \begin{bmatrix}
        \dot{\phi}_1 (t) \\
        \dot{\phi}_2 (t)
    \end{bmatrix}
    &= -\lambda
    \begin{bmatrix}
        \phi_1 (t) \\
        \phi_2 (t)
    \end{bmatrix}
    +
    \begin{bmatrix}
        x_p (t) \\
        u (t)
    \end{bmatrix}, \\
    \dot{\hat{\bar{w}}} \left( t \right) &= \proj_{\Pi,\hat{\bar{w}}} \left( -\Gamma \left( E^\top \left( t \right) E \left( t \right) \hat{\bar{w}} \left( t \right) + E^\top \left( t \right) \varepsilon_{5} \left( t \right) \right) \right), \\
    \hat{w}_{5} \left( t \right) &= 1 - \sum_{i=1}^{4} \hat{w}_i \left( t \right), \\
    \hat{K} \left( t \right) &= \hat{B}_p^{\dagger} (t) \sum_{i=1}^{5} \hat{w}_i \left( t \right) B_i K_i, \\
    \hat{L} \left( t \right) &= \hat{B}_p^{\dagger} (t) \sum_{i=1}^{5} \hat{w}_i \left( t \right) B_i L_i,  \\
    u \left( t \right) &= \hat{K} \left( t \right) x_p \left( t \right) + \hat{L} \left( t \right) r \left( t \right).
\end{align*}

The simulations of the MMRAC are compared to a direct adaptive control technique using a single model MRAC (see Chapter 9 of~\cite{tao2003adaptive}). The simulations consider the same initial condition for the system, the same reference, and the initial gains are calculated as
\begin{align*}
    \hat{K}(0) = \hat{B}_p^{\dagger} (0)\sum_{i=0}^{5} w_{i}(0) K_i, \\
    \hat{L}(0) = \hat{B}_p^{\dagger} (0)\sum_{i=0}^{5} w_{i}(0) L_i.
\end{align*}

In~\Cref{fig:track} we see the time evolution of each of the states for the reference model, the states generated by the MMRAC, and the states generated by MRAC. Convergence is achieved for both controllers, with a clear advantage of the multiple-model technique while using approximately the same control effort, as depicted in~\Cref{fig:controleffort}. The convergence speed is a key factor to consider, and there is a clear advantage of using multiple models when we analyze the dynamics of the error. In~\Cref{fig:errornorm} we compare the time evolution of the norm of the error for the case of MMRAC, and MRAC. To better visualize the advantage in convergence speed, we present the norm of the error in log-scale in~\Cref{fig:convergencespeed}, as we perform a linear regression on each of the signals. We see that the convergence speed of the MMRAC scheme is approximately two order of magnitude faster than the single model approach (the slope of the linear regressions are -0.0333 and -0.0103 for MMRAC and a single model, respectively), matching with the conjecture in~\cite{NarendraKumpatiS2014Srap}. Finally, as mentioned in~\Cref{sec:sysidentification}, we achieve asymptotic convergence of the estimated system's matrices. In~\Cref{fig:ApSimulation,fig:BpSimulation} we see the asymptotic convergence of each entry of the matrix pairs $\begin{bmatrix}
    \hat{A}_p(t) & \hat{B}_p(t)
\end{bmatrix}$ to $\begin{bmatrix}
    A_p & B_p
\end{bmatrix}$, respectively.

\section{Conclusions}
\label{sec:Conclusions}

In this article we developed a multiple model reference adaptive control (MMRAC) scheme for multi-input, linear, time-invariant systems with uncertain parameters that lie inside a known compact and convex polytope. The controller performs online parameter identification of the system matrices as a convex combination of an arbitrary number of fixed model, one at each extreme point of the convex polytope of uncertainty. The identification is proven to be asymptotically stable, and sufficient conditions for perfect identification of the uncertain system matrices are provided. The tracking controller of the MMRAC guarantees that all close-loop signals are bounded, and that the difference between the plant states and the ones generated by a linear reference model asymptotically converge to zero. To verify the effectiveness of the proposed MMRAC scheme we compare it to a single model MRAC scheme through MATLAB and Simulink simulations. For the simulation examples, we observe that the convergence speed is approximately two orders of magnitude faster than using a single model, with similar control effort. Future lines of research include generalizing the proposed scheme to linear, time-varying systems and nonlinear systems. 

\bibliographystyle{IEEEtran}
\bibliography{IEEEfull,root}

\end{document}